\documentclass[11pt]{article}

\usepackage{nicefrac}

\usepackage[ruled]{algorithm2e}

\SetAlFnt{\small}
\SetAlCapFnt{\small}
\SetAlCapNameFnt{\small}
\SetAlCapHSkip{0pt}
\IncMargin{-\parindent}

\usepackage[numbers]{natbib}
\usepackage{amsmath, amssymb, nicefrac, comment}
\usepackage[numbers]{natbib}
\usepackage{algorithmic,bm,latexsym,amsmath,amssymb,graphicx,mathrsfs,comment,graphics,nicefrac,breqn}

\newtheorem{definition}{Definition}

\newcommand{\MC}{\mathit{MC}}
\newcommand{\C}{\mathcal{C}}
\newcommand{\G}{\mathscr{G}}
\newcommand{\N}{\mathscr{N}}

\newcommand{\eqComment}[1]{\textnormal{\fontsize{8}{8}\selectfont{\ \ \ (#1)}}}

\newtheorem{theorem}{Theorem}[section]
\newtheorem{claim}[theorem]{Claim}

\newtheorem{lemma}[theorem]{Lemma}
\newtheorem{example}[theorem]{Example}
\newtheorem{proof}[theorem]{Proof}

\begin{document}

{\begin{center}
\fontsize{15}{15}{\selectfont{\textbf{A Measure of Synergy in Coalitions}}}
~\\~\\~\\
\fontsize{11}{11}{\selectfont{
Talal Rahwan$^1$, Tomasz Michalak$^2$ and Michael Wooldridge
~\\~\\
\fontsize{10}{10}{\selectfont{
$^1$Masdar Institute of Science and Technology, UAE
~\\
$^2$Department of Computer Science, University of Oxford, UK
}}
~\\~\\
(11 April, 2014)
~\\~\\~\\~\\
\textbf{Abstract}
~\\~\\
}}
\fontsize{11}{11}{\selectfont{
\parbox{12cm}{
  When the performance of a team of agents exceeds our expectations or
  fall short of them, we often explain this by saying that there was
  some \emph{synergy} in the team---either positive (the team exceeded
  our expectations) or negative (they fell short). Our aim in this
  article is to develop a formal and principled way of measuring
  synergies, both positive and negative. Using characteristic function
  cooperative games as our underlying model, we present a formal
  measure of synergy, based on the idea that \emph{a synergy is
    exhibited when the performance of a team deviates from the norm}.
  We then show that our synergy value is the only possible such
  measure that satisfies certain intuitive properties. We then
  investigate some alternative characterisations of this measure.
  }}}
\end{center}}
%

\section{Introduction}\label{sec:intro}

According to the Oxford English dictionary, synergy is \textit{``the
  interaction or cooperation of two or more organizations, substances,
  or other agents to produce a combined effect greater than the sum of
  their separate effects''}. Here, we interpret the word ``greater''
to mean either greater in harm, or in benefit. In other words, a
synergy can be either negative, positive, or zero. We address the
following fundamental question: \textit{what is an appropriate measure
  of synergy?} For example, suppose we have a set of school
children, who are given various group assignments in different subjects
throughout the year. At the end of the year, every possible group (or
``coalition'' as we call it) has worked on some activity, for which
they have received a grade. Now, suppose the teacher is asked to
characterise the synergies of various different groups---the extent to
which the groups work particularly well or particularly badly. What is
a reasoned way to do this? It is this question that we address in the
present paper.

To illustrate some of the subtleties involved in this issue, suppose
the student groups $C = \{1, 2, 3\}$ and $C' = \{4, 5, 6\}$ each
receive a grade of, say, 70 each. Naively it might appear that they
exhibit equal synergies. But suppose:

\begin{itemize}
\item Each of the three student in $C$ is shy. This may inhibit
  him/her from making any substantive contribution in a group
  activity. However, suppose those three students are good friends. As
  such, when they were all together in a coalition with no one else,
  they could overcome their shyness. Consequently, they achieved a
  score of 70---a score that indicates an improvement in the overall
  performance (considering the small contribution that each of them
  usually makes). Based on this, the teacher may conclude that there
  is a \textit{positive synergy} among those students; putting them
  together in a group helps bring the best out of them.
\item Each student in $C'$ has good leadership skills. As such, they
  each usually perform very well in group activities. However, when
  they are in the same coalition with no one else, they argue
  constantly. As a result, they underperform. For those students, the
  teacher may consider a score of 70 to be relatively low; the teacher
  expects more of them. This suggests that those students must not be
  placed together in a group with no one else, otherwise there will be
  a \textit{negative synergy} among them.
\end{itemize}
So, although $C$ and $C'$ have the same grade, the synergy seems to be
positive in the first coalition, and negative in the second. This
suggests that the grade alone might not reveal the true nature of the
synergy among the members of a coalition.

An alternative way to quantify synergy is to think of it as a
\textit{surplus}---the \textit{extra} score obtained through the
formation of the group. This surplus can be quantified, for example,
by comparing the score of the group with the scores obtained when each
member works alone.\footnote{\footnotesize Other ways to quantify the
  surplus are discussed in Section~\ref{sec:relatedWork}.} Here, the
intuition is that the performance of any given member is expected to
be equal to his or her performance when working alone; any deviation
from this expectation is then attributed to the synergy among the
group members. Based on this, the suitability of such a measure
depends on the answer to the following fundamental question:

\vspace*{0.5ex}\hspace*{0.35cm}\parbox{12.5cm}{\noindent\textit{If the
    performance of an individual working alone is different than his
    or her performance when working in a certain group, should this
    difference be entirely attributed to that group, or is the
    individual also (at least partially) responsible?
  }}\vspace*{0.5ex}

\noindent Note that with the aforementioned interpretation of synergy,
the entire difference is indeed attributed to the group. We argue that
this may not necessarily be the most appropriate answer in all
cases. In our school scenario, for example, there could be a student
who is intelligent but socially anxious. This student excels when
working alone, but underperforms in group activities, regardless of
the identities of his or her team mates. This suggests that any such
deterioration in the student's performance is not necessarily the
group's fault.

The above discussion suggests the need for a more sophisticated
measure that inherently avoids the aforementioned limitations.  We
propose a measure of synergy that, instead of considering it to be a
\textit{surplus}, considers to be a \textit{deviation from the
  norm}. In other words, instead of being something \textit{extra}, we
consider it to be something \textit{special} or \textit{abnormal}. A
\textit{measure of synergy} then reflects the degree to which a
certain coalition differs from the norm. One can then define the norm
such that the resulting measure exhibits certain desirable
properties. Such properties could be aimed, for example, at capturing
the way in which each individual normally behaves in different groups,
e.g., to identify any students who are social or antisocial by
nature. With this in mind, we propose one such measure, and show that
it is the only possible measure that satisfies certain intuitive
properties.

\section{Preliminaries}\label{sec:preliminaries}

Throughout the paper, we will use standard notions from cooperative
game theory. In particular, a \textit{characteristic-function game} is
a pair, $(N,v)$, where $N=\{1,\dots,n\}$ is the set of
\textit{players}, and $v:2^N \to {\mathbb R}$ is the
\textit{characteristic function}, which assigns to each subset (or
\textit{``coalition''}) $C\subseteq N$ a real value, $v(C)$---called
\textit{the value of $C$}---which can be interpreted as the utility
attainable by that coalition. For instance, in a game that represents
our school scenario, every player represents a student, while the
characteristic function returns the grade attained by a group.  As
common in the literature, we will assume that $v(\emptyset) = 0$.  We
will denote the set of all such games by $\G$, and will denote the
universe of players by $\N$.

The following standard definitions will be used heavily throughout the
paper:

\begin{definition}[\textbf{Symmetric Players}]\label{def:symmetricPlayers}
Let $(N,v)$ be a game, and let $i,j\in N$. Players $i$ and $j$ are said to be \emph{symmetric} if and only if: $v(C\cup\{i\}) = v(C\cup\{j\})$ for every $C\subseteq N\setminus\{i,j\}$.
\end{definition}

From a strategic perspective, any two symmetric players are identical in the game; they only differ in their names.

\begin{definition}[\textbf{Null Player}]\label{def:nullPlayer}
Let $(N,v)$ be a game. A player $i\in N$ is called a \emph{null player} if and only if $v(C\cup\{i\})=v(C)$ for all $C\subseteq N$.\footnote{\footnotesize If a coalition $C$ contains player $i$ then the equality holds trivially, because then $C\cup\{i\} = C$.}
\end{definition}

In words, whenever a null player joins a coalition, it makes no impact on that coalition's value.

\begin{definition}[\textbf{Dummy Player}]\label{def:dummyPlayer}
Let $(N,v)$ be a game. A player $i\in N$ is called a \emph{dummy player} if and only if $v(C\cup\{i\}) = v(C)+v(\{i\})$ for all $C\subseteq N\setminus\{i\}$.
\end{definition}

Intuitively, a dummy player is one whose performance is never affected
by his or her team mates; the player always produces exactly the same
amount of utility. We will say a game $(N,v)$ is a \emph{game of dummies} if every player in it is a dummy player.

\begin{definition}[\textbf{Sum of Two Games}]\label{def:sumOfTwoGames}
  The \emph{sum of two games}, $(N,v)$ and $(N,w)$, is a game denoted
  by $(N,v+w)$, where $(v+w)(C) = v(C)+w(C)$ for every coalition
  $C\subseteq N$.
\end{definition}

Similarly, subtracting $(N,v)$ from $(N,w)$ results in a game denoted
by $(N,v-w)$, where $(v-w)(C) = v(C)-w(C)$ for every coalition
$C\subseteq N$. 

\begin{definition}[\textbf{Shapley Value}]\label{def:shapleyValue}
Let $(N,v)$ be a game, and let $i\in N$. The \emph{Shapley value}~\cite{Shapley:53} of $i$ is denoted by $\phi_i(N,v)$ and given by:
\begin{equation}\label{eqn:shapleyValue}
\phi_i(N,v) = \sum_{C \subseteq N\setminus \{i\}} \frac{|C|!\ (n-|C|-1)!} {n!} \big( v(C\cup\{i\}) - v(C) \big)
\end{equation}
\end{definition}


Assuming that the grand coalition (i.e., $N$) will be formed, and that
we want to divide $v(N)$ among the players, Shapley showed that
$(\phi_1(N,v), \dots, \phi_n(N,v))$ is the only possible division of
$v(N)$ that satisfies all of the following properties:

\begin{itemize}
    \item[] $P_1$. (\textbf{Symmetry)}: For every game $(N,v)\in\G$ and every pair of symmetric players $i$ and $j$ in the game: $\phi_i(N,v)=\phi_j(N,v)$;\smallskip
    \item[] $P_2$. (\textbf{Null Player)}: For every game $(N,v)\in\G$, and every null player $i$ in the game: $\phi_i(N,v) = 0$;\smallskip
    \item[] $P_3$. (\textbf{Efficiency)}: For every game $(N,v)\in\G$, we have: $\sum_{i\in N}\phi_i(N,v) = v(N)$;\smallskip
    \item[] $P_4$. (\textbf{Additivity)}: For every pair of games, $(N,v)$ and $(N,w)$, and every player $i\in N$, we have: $\phi_i(N,v+w) = \phi_i(N,v) + \phi_i(N,w)$.\smallskip
\end{itemize}
Properties $P_1$---$P_4$ are commonly known as the Shapley axioms.


\section{A New Measure of Synergy}\label{sec:FirstDefinition}

We start with a general definition of what a synergy measure actually is.

\begin{definition}[\textbf{Measure of Synergy}]\label{def:measureOfSynergy}
A measure of synergy $\psi$ is a real valued function on:
$$
\bigcup_{(N,v)\in \G} \big\{(N,v,C) : C\subseteq N\big\},
$$
where $\psi((N,v,C))$ represents the amount of synergy among the members in $C$ in the game $(N,v)$. The set of all such measures will be denoted by $\Psi$.
\end{definition}

For notational convenience, we will write $\psi^N_v(C)$ instead of $\psi((N,v,C))$. Furthermore, since it is typically assumed that the empty set has no value (i.e., $v(\emptyset)=0$), it makes sense to also assume it has no synergy. As such, we will assume that $\psi^N_v(\emptyset) = 0$ for all $(N,v)\in\G$.

To try and identify a reasonable measure in $\Psi$, we considered the following properties, which are inspired by the Shapley axioms:

\begin{itemize}
    \item[] $P_5$. (\textbf{Symmetric-Synergy}): A measure $\psi\in\Psi$ satisfies $P_5$ if, for every game $(N,v)\in\G$, and every pair of symmetric players $i$ and $j$ in that game, we have: $\psi^N_v(C\cup\{i\})=\psi^N_v(C\cup\{j\})$ for all $C\subseteq N\setminus\{i,j\}$;\smallskip
    \item[] $P_6$. (\textbf{Null-Synergy}): A measure $\psi\in\Psi$ satisfies $P_6$ if $\psi^N_v(C) = \psi^N_v(C\cup\{i\})$ for every game $(N,v)$ and every null player $i$ in that game;\smallskip
    \item[] $P_7$. (\textbf{Dummy-Synergy}): A measure $\psi\in\Psi$ satisfies $P_7$ if, for every game $(N,v)\in\G$, the following is a game of dummies: $(N,v-\psi^N_v)$;
    \item[] $P_8$. (\textbf{Normalized-Synergy}): A measure $\psi\in\Psi$ satisfies $P_8$ if, for every game $(N,v)\in\G$, we have: $\sum_{C\subseteq N}\psi^N_v(C) = 0$;\smallskip
    \item[] $P_9$. (\textbf{Additive-Synergy}): A measure $\psi\in\Psi$ satisfies $P_9$ if, for every pair of games, $(N,v)$ and $(N,w)$, and every coalition $C\subseteq N$, we have: $\psi^N_v(C) + \psi^N_w(C) = \psi^N_{v+w}(C)$.\smallskip
\end{itemize}
Observe that, while $P_5$, $P_6$, and $P_8$ are fairly natural
counterparts of Shapley axioms, properties $P_7$ and $P_8$ are
new. Let us discuss and motivate these properties in more detail.

Property~$P_5$ states that \textit{replacing a player with another
  symmetric one does not alter the synergy}. This seems natural,
because any two symmetric players are essentially identical except for
having different names. Thus, it seems strange to claim that they
differ in the way they affect the synergies.

Property~$P_6$ states that \textit{a null player does not affect the
  synergy in a coalition}. The intuition behind this may become
clearer through the following example:

\begin{example} Let $(N,u)$ be a game where $N=\{1,2,3\}$ and $u$ is
  defined as follows:
\begin{center}
\begin{tabular}{lll}
  $u(\emptyset) = 0$ & $u(\{3\}) = 0$ \\
  $u(\{1\}) = 10$  & $u(\{1,3\}) = 10$ \\
  $u(\{2\}) = 20$  & $u(\{2,3\}) = 20$ \\
  $u(\{1,2\}) = 1000000$  & $u(\{1,2,3\}) = 1000000$ \\
\end{tabular}
\end{center}
\end{example}

Observe that, in this example, player 3 is a null player. Also see how
the utility attainable by 1 and 2 increases significantly when they
form a coalition (compared to the case where they each form a
singleton coalition). Looking at this game, it seems reasonable to
claim that whatever positive synergy there is between 1 and 2---which
made them achieve the million units of utility---is retained when
player 3 joins their coalition. In other words, the null player does
not affect the synergy when joining the coalition, just as stated by
Property~$P_6$.

Property $P_7$ essentially states that \textit{without synergies, we
  are left with a game in which every player is a dummy}, i.e., we are
left with an ``additive'' game---one in which the value of every
coalition $C\subseteq N$ is simply the sum of the values that the
members achieve when each of them $i\in C$ forms its own coalition
$\{i\}$. To make the underlying concept clearer, let us put it in the
context of our school scenario. In an ``additive'' classroom, whenever
a group is formed, each member completely ignores his or her team
mates, and acts exactly just like he or she would act when working
alone. It seems reasonable to claim that this is how a classroom would
look like if we strip away all synergies from it, so to speak.

Property~$P_8$ relates to our interpretation of synergy as a
\textit{deviation from the norm}. This interpretation implies that a
coalition whose value is greater than the norm has positive synergy,
while a coalition whose value is lower than the norm has negative
synergy. More importantly, a coalition whose value follows the norm
has exactly zero synergy. Perhaps the most intuitive interpretation of
the ``\textit{norm}'' is to consider it to be the ``\textit{average}''
case. As such, it seems reasonable that the coalitions on average
follow the norm. That is to say, it seems reasonable that synergy on
average equals zero, just as stated in Property~$P_8$.

Finally, let us comment on $P_9$. While this property is admittedly
not very intuitive, at least it has an intuitive implication: Scaling
a game---by multiplying all coalition values by some constant---does
not change the relative differences between the synergies in that
game.

Taking the above properties into consideration, we propose the
following measure of synergy, called the \textit{synergy value}:

\begin{definition}[\textbf{Synergy Value}]\label{def:synergyValue}
  The \emph{synergy value} is the measure of synergy, $\chi\in\Psi$,
  defined for every game $(N,v)\in\G$ and every coalition $C\subseteq
  N$ as follows:
\begin{equation}\label{eqn:synergyValue}
\chi^N_v(C) = v(C) - \sum_{i\in C}\overline{\phi}_i(N,v),
\end{equation}
where $\overline{\phi}_i(N,v)$ is the average Shapley value of player
$i$ taken over all sub-games $(C,v):C\subseteq N$. That is,
\begin{equation}\label{eqn:averageShapleyValue}
\overline{\phi}_i(N,v) = 2^{1-n} \sum_{C\subseteq N : i\in C} \phi_i(C,v).
\end{equation}
\end{definition}

We conclude this section with the following theorem, which essentially states that if we view properties $P_5$ to $P_9$ as axioms, then those axioms characterize the Synergy value $\chi$.

\begin{theorem}\label{thm:firstAxiomatizationOfSynergyValue}
  The synergy value $\chi$ is the only measure in $\Psi$ that
  satisfies properties $P_5$ to $P_9$.
\end{theorem}

\begin{proof}
  We begin by showing that the synergy value satisfies the five
  properties listed in the statement of the theorem.
\begin{claim}\label{clm:proof:firstAxiomatizationOfSynergyValue}
The Synergy value $\chi$ satisfies properties $P_5$ to $P_9$.
\end{claim}
\begin{proof}
Let $(N,v)$ be a game. To prove that $\chi$ satisfies $P_5$, for every pair of symmetric players $i$ and $j$ in the game, and every coalition $C\subseteq N\setminus\{i,j\}$, we need to prove that: $\chi^N_v(C\cup\{i\}) = \chi^N_v(C\cup\{j\})$, i.e.,
\begin{equation}\label{proof:P5:1}
v(C\cup\{i\}) - \sum_{k\in C\cup\{i\}}\overline{\phi}_k(N,v)\ \ =\ \ v(C\cup\{j\}) - \sum_{k\in C\cup\{j\}}\overline{\phi}_k(N,v).
\end{equation}
Since $i$ and $j$ are symmetric, then: $v(C\cup\{i\}) = v(C\cup\{j\}), \forall C\subseteq N\setminus\{i,j\}$. Consequently, to prove that Equation~\eqref{proof:P5:1} holds, it suffices to show that:
\begin{equation}\label{proof:P5:2}
\overline{\phi}_i(N,v) = \overline{\phi}_j(N,v).
\end{equation}
To this end, observe that:
\begin{itemize}
\item For every $C\subseteq N: i,j\in C$, the players $i$ and $j$ are symmetric in the game $(C,v)$, and so $\phi_i(C,v)=\phi_j(C,v)$.
\item For every $C\subseteq N\setminus\{i,j\}$, we know from \eqref{eqn:shapleyValue} that:
\begin{equation}\label{proof:P5:3}
\phi_i(C\cup\{i\},v) = \sum_{S \subseteq C} \frac {|S|!\ ((|C|+1)-|S|-1)!} {(|C|+1)!} \big( v(S\cup\{i\}) - v(S) \big),
\end{equation}
\begin{equation}\label{proof:P5:4}
\phi_i(C\cup\{j\},v) = \sum_{S \subseteq C} \frac {|S|!\ ((|C|+1)-|S|-1)!} {(|C|+1)!} \big( v(S\cup\{j\}) - v(S) \big).
\end{equation}
Now since $i$ and $j$ are symmetric, and since every $S\subseteq C$ does not contain $i$ nor $j$, then by definition of symmetry: $v(S\cup\{i\}) = v(S\cup\{j\})$. Thus, equations \eqref{proof:P5:3} and \eqref{proof:P5:4} imply that $\phi_i(C\cup\{i\},v)=\phi_j(C\cup\{j\},v)$.
\end{itemize}

\noindent The above two cases imply that Equation~\eqref{proof:P5:2} always holds, meaning that $\chi$ satisfies $P_5$.

Moving on to $P_6$, we need to prove that $\chi^N_v(C\cup\{i\}) = \chi^N_v(C)$ for every null player $i$ in game $(N,v)$. Since $i$ is a null player in $(N,v)$, then it is also a null player in every sub-game $(C,v):C\subseteq N$, and so $\overline{\phi}_i(N,v)=0$. We also know from the definition of a null player that $v(C\cup\{i\}) = v(C)$ for every $C\subseteq N$. As such,
$$
\chi^N_v(C\cup\{i\}) = v(C\cup\{i\}) - \sum_{j\in C\cup\{i\}}\overline{\phi}_j(N,v) = v(C) - \sum_{j\in C}\overline{\phi}_j(N,v) = \chi^N_v(C).
$$
Next, we prove that $\chi$ satisfies $P_7$. Let $\vartheta = v-\chi^N_v$. That is, $\vartheta = \sum_{i\in C}\overline{\phi}_i(N,v)$. We need to prove that every player $i$ in game $(N,\vartheta)$ is a dummy. To this end, it suffices to note that for every $i\in N$ and $C\subseteq N\setminus\{i\}$:
$$
\vartheta(C\cup\{i\}) = \sum_{j\in C\cup\{i\}}\overline{\phi}_j(N,v) = \overline{\phi}_i(N,v) + \sum_{j\in C}\overline{\phi}_j(N,v) = \vartheta(\{i\})+\vartheta(C).
$$
Moving on to $P_8$, we need to show that the following holds:
\begin{equation}
\sum_{C\subseteq N} \chi^N_v(C) = 0.
\end{equation}
This can be shown as follows:
\begin{eqnarray}
\sum_{C\subseteq N} \chi^N_v(C)  &=& \sum_{C\subseteq N} \Bigg( v(C) - \sum_{i\in C} \overline{\phi}_i(N,v) \Bigg) \nonumber \\
                                 &=& \sum_{C\subseteq N} v(C)\ - 2^{1-n}\ \sum_{C\subseteq N}\ \sum_{i\in C}\ \sum_{S\subseteq N : i\in S} \phi_i(S,v) \nonumber \\
                                 &=& \sum_{C\subseteq N} v(C)\ - \sum_{i\in N}\ \sum_{S\subseteq N : i\in S} \phi_i(S,v) \nonumber \\
                                 &=& \sum_{C\subseteq N} v(C)\ - \sum_{S\subseteq N} \sum_{i\in S}\phi_i(S,v) \nonumber \\
                                 &=& \sum_{C\subseteq N} v(C)\ - \sum_{S\subseteq N} v(S)\ =\ 0. \nonumber
\end{eqnarray}
Finally, we deal with $P_9$. We need to prove that for every pair of games, $(N,v)$ and $(N,w)$, and every $C\subseteq N$, we have: $\chi^N_v(C) + \chi^N_w(C) = \chi^N_{v+w}(C)$. Since $\phi$ satisfies $P_4$, it follows that $\overline{\phi}_j(N,v+w) = \overline{\phi}_j(N,v) + \overline{\phi}_j(N,w)$. Thus,
\begin{eqnarray}
\chi^N_{v+w}(C\cup\{i\}) &=& (v+w)(C\cup\{i\})\ - \sum_{j\in C\cup\{i\}}\overline{\phi}_j(N,{v+w}) \nonumber \\
                         &=& v(C)\ - \sum_{j\in C}\overline{\phi}_j(N,v)\ +\ w(C)\ +\sum_{j\in C}\overline{\phi}_j(N,w) \nonumber \\
                         &=& \chi^N_v(C) + \chi^N_w(C). \nonumber
\end{eqnarray}
This concludes the proof of Claim~\ref{clm:proof:firstAxiomatizationOfSynergyValue}.
\end{proof}

Having proved that $\chi$ satisfies properties $P_5$ to $P_9$, it remains to prove that $\chi$ is in fact the only measure in $\Psi$ that satisfies those properties. In other words, assuming that $x\in\Psi$ is a measure satisfying properties $P_5$ to $P_9$, we need to show that $x=\chi$. This will be done using of the following lemma.

\begin{lemma}\label{lem:proof:uniquenessOfSyergyValue}
Given a game $(N,v)$, a coalition $S\subseteq N$, and a constant $\alpha\in{\mathbb R}$, let us denote by $(N,v_{S,\alpha})$ the game where the value of a coalition $C\subseteq N$ is:
\begin{equation}\label{eqn:lem:proof:uniquenessOfSyergyValue:1}
v_{S,\alpha}(C) = \begin{cases}
\alpha &\ \ \text{if } S\subseteq C,\\
0 &\ \ \text{otherwise.}
\end{cases}
\end{equation}
If a measure of synergy $\psi\in\Psi$ satisfies properties $P_5$, $P_6$, $P_7$ and $P_8$, then
$$
\psi^N_{v_{S,\alpha}}(C) = \begin{cases}
(1-2^{1-|S|})\ \alpha &\ \ \ \text{if } S\subseteq C,\\
-|C\cap S|\ (2^{1-|S|}) \frac{\alpha}{|S|} &\ \ \ \text{otherwise.}
\end{cases}
$$
\end{lemma}

\begin{proof}
In the game $(N,v_{S,\alpha})$, Property~$P_7$ states that $(N,v_{S,\alpha}-\psi^N_{v_{S,\alpha}})$ is a game of dummies. In other words, there exist real numbers $(\beta_i)_{i\in N}$ such that:
\begin{equation}\label{eqn:lem:proof:uniquenessOfSyergyValue:2}
\forall C \subseteq N,\ \  v_{S,\alpha}(C) - \psi^N_{v_{S,\alpha}}(C) = \sum_{i\in C} \beta_i.
\end{equation}
Now, observe that every $i\in N\setminus S$ is a null player, meaning that: $v_{S,\alpha}(\{i\})=0$ and that $\psi^N_{v_{S,\alpha}}(\{i\})=0$ (based on Property~$P_6$). Thus, based on Equation~\eqref{eqn:lem:proof:uniquenessOfSyergyValue:2}, we have:
%
%
\begin{equation}\label{eqn:lem:proof:uniquenessOfSyergyValue:3}
\beta_i = 0, \ \ \forall i\in N\setminus S.
\end{equation}
Furthermore, observe that every pair of players $i,j\in S$ are symmetric, meaning that $v_{S,\alpha}(\{i\}) = v_{S,\alpha}(\{j\})$ and that $\psi^N_{v_{S,\alpha}}(\{i\})=\psi^N_{v_{S,\alpha}}(\{j\})$ (based on Property~$P_5$). This, as well as Equation~\eqref{eqn:lem:proof:uniquenessOfSyergyValue:2}, imply that there exists a real number, $\beta\in{\mathbb R}$, such that:
%
%
\begin{equation}\label{eqn:lem:proof:uniquenessOfSyergyValue:4}
\beta_i = \beta, \forall i\in S.
\end{equation}
Property~$P_8$, as well as equations \eqref{eqn:lem:proof:uniquenessOfSyergyValue:1} to \eqref{eqn:lem:proof:uniquenessOfSyergyValue:4}, imply that:
$$
\Bigg( \sum_{s=1}^{|S|-1} \sum_{C\subseteq N:|C\cap S|=s} (-s\beta) \Bigg) + \sum_{C\subseteq N:S\subseteq C} (\alpha - |S|\beta) = 0. 
$$
Thus, we have:
$$
\beta = \frac{\alpha 2^{1-|S|}}{|S|2^{n-1}}.
$$
This, as well as equations \eqref{eqn:lem:proof:uniquenessOfSyergyValue:1} to \eqref{eqn:lem:proof:uniquenessOfSyergyValue:4} imply the correctness of Lemma~\ref{lem:proof:uniquenessOfSyergyValue}.
\end{proof}

Recall that we wanted to use Lemma~\ref{lem:proof:uniquenessOfSyergyValue} to prove that $x=\chi$. To do this, we need to first introduce the notion of a \textit{``carrier game''}~\cite{Shapley:53}. In particular, for every coalition $S\subseteq N$, we will denote by $(N,v_{S})$ the \textit{carrier game over $S$}---the game in which the value of a coalition $C\subseteq N$ is:
$$
v_{S}(C) = \begin{cases}
1 &\ \ \text{if } S\subseteq C,\\
0 &\ \ \text{otherwise.}
\end{cases}
$$
Shapley~\cite{Shapley:53} proved that every game $(N,v)$ is a linear combination of carrier games. In other words, there exist real numbers $(\alpha_S)_{S\subseteq N, S\neq \emptyset}$ such that, for every $C\subseteq N$, we have:
\begin{equation}\label{eqn:carrier:1}
v(C) = \sum_{S\subseteq N, S\neq \emptyset} v_{S,\alpha_S}(C).
\end{equation}
Now since $x$ and $\chi$ both satisfy properties $P_5$, $P_6$, and $P_7$, then Lemma~\ref{lem:proof:uniquenessOfSyergyValue} implies that:
\begin{equation}\label{eqn:carrier:2}
x^N_{v_{S,\alpha_S}} = \chi^N_{v_{S,\alpha_S}}\ ,\ \ \ \ \forall S \subseteq N, S \neq \emptyset.
\end{equation}
Based on equations \eqref{eqn:carrier:1} and \eqref{eqn:carrier:2}, as well as the fact that both $x$ and $\chi$ satisfy Property~$P_9$, we have:
$$
x^N_{v} = \sum_{S\subseteq N, S\neq \emptyset} x^N_{v_{S,\alpha_S}} = \sum_{S\subseteq N, S\neq \emptyset}\chi^N_{v_{S,\alpha_S}} = \chi^N_{v}.
$$		
Since this is true for every game, $(N,v)\in\G$, we conclude that $x=\chi$. This concludes the proof of Theorem~\ref{thm:firstAxiomatizationOfSynergyValue}.
$\ \ \ \square$
\end{proof}


\section{An Alternative Axiomatization}\label{sec:AlternativeAxiomatization}

The characterization that we introduced in
Section~\ref{sec:FirstDefinition} for the Synergy value,
$\chi\in\Psi$, relies on Property~$P_9$, which is admittedly not very
intuitive. Therefore, in this section we identify an alternative
axiomatization of $\chi$ that replaces Property~$P_9$ with a somewhat
more intuitive property. First, let us introduce the following
definition from cooperative game theory:

\begin{definition}[\textbf{Marginal Contribution}]\label{def:marginalContribution}
Let $(N,v)$ be a game in $\G$. The marginal contribution of a player $i\in N$ to a coalition $C\subseteq N$ is denoted by $\MC^C_i(N,v)$, and is given by: $\MC^C_i(N,v) = v(C\cup\{i\})-v(C)$.
\end{definition}

Observe that, $\MC^C_i=0$ for all $i\in N$ and $C\subseteq N:i\in C$. Having defined the marginal contribution, we are now ready to introduce a new property, called $P_{10}$.

\begin{itemize}
    \item[] $P_{10}$. (\textbf{Marginal-Synergy}): A measure $\psi\in\Psi$ satisfies $P_{10}$ if, for every pair of games $(N,v)$ and $(N,w)$ with the same set of players, and for every coalition $C\subseteq N$, if
    \begin{equation}\label{eqn:marginalityInSynergy}
    \forall i\in C, \forall S\subseteq N\setminus \{i\},\ \ \MC^S_i(N,v) = \MC^S_i(N,w),
    \end{equation}
    then the following holds:
    $$
    \psi^N_v(C) = \psi^N_w(C).
    $$
\end{itemize}

To put this in the context of our school scenario, imagine that by the end of Year~1, the teacher has developed a certain opinion on the synergy among a certain group of students, $C\subseteq N$ (where $N$ is the set of students in the classroom). Furthermore, imagine that, in Year~2, the classroom consisted of the same set of students---$N$. Finally, imagine that the performance of every member of $C$ remained unchanged (compared to Year~1). Then, it seems reasonable for the teacher to keep her opinion regarding the synergy among the members of $C$, even if the performance (of some or all) of the students outside $C$ has changed (compared to Year~1). 

We conclude this section with Theorem~\ref{thm:secondAxiomatizationOfSynergyValue}, which essentially states that $P_6$, $P_7$, $P_8$, and $P_{10}$ are axioms that characterize the Synergy value $\chi$. Compared to the previous axiomatization from Section~\ref{sec:FirstDefinition}, here we replace properties $P_9$ and $P_5$ with Property~$P_{10}$.

\begin{theorem}\label{thm:secondAxiomatizationOfSynergyValue}
The Synergy value, $\chi\in\Psi$, is the only measure in $\Psi$ that satisfies properties $P_6$, $P_7$, $P_8$, and $P_{10}$.
\end{theorem}

\begin{proof}
We begin the proof of Theorem~\ref{thm:secondAxiomatizationOfSynergyValue} by ascertaining that the Synergy value satisfies the four properties listed in the statement of the theorem. As for $P_6$, $P_7$ and $P_8$, we already know from Claim~\ref{clm:proof:firstAxiomatizationOfSynergyValue} that $\chi$ satisfies those properties. Thus, we only need to prove that $\chi$ satisfies Property~$P_{10}$.

\begin{claim}\label{clm:syergyValueSatisfiesMarginality}
The Synergy value, $\chi\in\Psi$, satisfies property $P_{10}$.
\end{claim}

\begin{proof}
Let $(N,v)$ and $(N,w)$ be two games with the same set of players, $N$, and let $C\subseteq N$ be a coalition for which Equation~\eqref{eqn:marginalityInSynergy} holds. Based on this, it is easy to show that the following two equations hold:\footnote{\footnotesize The difference between equations \eqref{eqn:marginalityInSynergy} and \eqref{eqn:syergyValueSatisfiesMarginality:2} is that the former deals with the games $(N,v)$ and $(N,w)$, while the latter deals with all sub-games $(T,v)$ and $(T,w)$ such that $T\subseteq N:i\in T$.}
\begin{equation}\label{eqn:syergyValueSatisfiesMarginality:1}
v(C) = w(C).
\end{equation}
\begin{equation}\label{eqn:syergyValueSatisfiesMarginality:2}
\forall i\in C,\ \forall T\subseteq N:i\in T,\ \forall S\subseteq T\setminus \{i\},\ \ \ \MC^S_i(T,v) = \MC^S_i(T,w).
\end{equation}
Thus, based on the definition of the average Shapley value, $\overline{\phi}_i$, we find that:
\begin{equation}\label{eqn:syergyValueSatisfiesMarginality:3}
\forall i\in C,\ \overline{\phi}_i(N,v) = \overline{\phi}_i(N,w).
\end{equation}
Equations \eqref{eqn:syergyValueSatisfiesMarginality:1} and \eqref{eqn:syergyValueSatisfiesMarginality:3} as well as Definition~\ref{def:synergyValue}---the definition of the Synergy value---immediately imply the correctness of Claim~\ref{clm:syergyValueSatisfiesMarginality}.
\end{proof}

Having proved Claim~\ref{clm:syergyValueSatisfiesMarginality}, all that remains is to prove the uniqueness statement in Theorem~\ref{thm:secondAxiomatizationOfSynergyValue}. For this, we need to first introduce two new definitions, and to prove yet another claim.
\begin{definition}[$N^+_v$]\label{def:N^+_v}
For every game $(N,v)\in\G$, the set $N^+_v \subseteq N$ is the set of players in $(N,v)$ who are \emph{not} null players. More formally,
\begin{equation}\label{eqn:N^+}
N^+_v = \{i\in N: \exists C\subseteq N, \MC^C_i(N,v)\neq 0\}.
\end{equation}
\end{definition}

\begin{definition}[$(N,v_i)$]\label{def:(N,v_i)}
For every set of players, $N\subset\N$, and every $i\in N$, and every characteristic function $v:2^N \to {\mathbb R}$, the game $(N,v_i)$ is defined as follows:
\begin{equation}\label{eqn:v_i}
v_i(C) = v(C\setminus\{i\}),\ \ \forall C\subseteq N.
\end{equation}
\end{definition}
\begin{claim}\label{clm:proof:secondAxiomatizationOfSynergyValue}
For every game $(N,v)\in\G$, and every $i\in N^+_v$, the following holds:
$$
N^+_{v_i} \subset N^+_v.
$$
\end{claim}

\begin{proof}
Let $i$ be a player in $N^+_v$. Based on Definition~\ref{def:(N,v_i)}, we have:
%
%

\begin{eqnarray}\label{eqn:clm:proof:secondAxiomatizationOfSynergyValue:1}
\forall C\subseteq N: i\in C,\ \ \MC^C_i(N,v_i) &=& v_i(C\cup\{i\}) - v_i(C) \eqComment{based on Definition~\ref{def:marginalContribution}} \nonumber\\
                                                &=& v_i(C) - v_i(C) \eqComment{because $i\in C$} \nonumber\\
                                                &=& 0.
\end{eqnarray}
\begin{eqnarray}\label{eqn:clm:proof:secondAxiomatizationOfSynergyValue:2}
\forall C\subseteq N: i\notin C,\ \ \MC^C_i(N,v_i) &=& v_i(C\cup\{i\}) - v_i(C) \eqComment{based on Definition~\ref{def:marginalContribution}} \nonumber\\
                                                      &=& v(C\cup\{i\}\setminus\{i\}) - v(C\setminus\{i\})  \eqComment{based on Definition~\ref{def:(N,v_i)}} \nonumber\\
                                                      &=& v(C)-v(C)  \eqComment{because $i\notin C$} \nonumber\\
                                                      &=& 0. 
\end{eqnarray}
Equations \eqref{eqn:clm:proof:secondAxiomatizationOfSynergyValue:1} and \eqref{eqn:clm:proof:secondAxiomatizationOfSynergyValue:2} imply that $i$ is a null player in the game $(N,v_i)$. Thus, based on Definition~\ref{def:N^+_v}:
\begin{equation}\label{eqn:clm:proof:secondAxiomatizationOfSynergyValue:3}
i\notin N^+_{v_i}.
\end{equation}
On the other hand, the following always holds for every game $(N,v)$ and any two distinct player, $y,z\in N:y\neq z$:
\begin{eqnarray}\label{eqn:clm:proof:secondAxiomatizationOfSynergyValue:4}
\forall C\subseteq N,\ \MC^C_z(N,v_y) &=& v_y(C\cup \{z\}) - v_y(C) \nonumber\\
                                      &=& v((C\cup \{z\})\setminus\{y\}) - v(C\setminus\{y\}) \nonumber\\
                                      &=& v((C\setminus\{y\})\cup \{z\}) - v(C\setminus\{y\}) \nonumber\\
                                      &=& \MC^{C\setminus\{y\}}_z(N,v).
\end{eqnarray}
For every $j\in N\setminus N^+_v$, we know that $j\neq i$ (because $i\in N^+_v$). Then, based on Equation~\eqref{eqn:clm:proof:secondAxiomatizationOfSynergyValue:4}, we have:
\begin{equation}\label{eqn:clm:proof:secondAxiomatizationOfSynergyValue:5}
\forall j\in N\setminus N^+_v,\ \ \MC^C_j(N,v_i) = \MC^{C\setminus\{i\}}_j(N,v).
\end{equation}
Furthermore, for every $j\in N\setminus N^+_v$, we know from Definition~\ref{def:N^+_v} that $j$ is a null player in $(N,v)$. This, as well as Equation~ \eqref{eqn:clm:proof:secondAxiomatizationOfSynergyValue:5}, imply that
$$
\MC^C_j(N,v_i) = 0,\ \ \forall j\in N\setminus N^+_v.
$$
This means that $j$ is also a null player in $(N,v_i)$, not just in $(N,v)$. Therefore, based on Definition~\ref{def:N^+_v}, we have:
\begin{equation}\label{eqn:clm:proof:secondAxiomatizationOfSynergyValue:6}
j\notin N^+_{v_i},\ \ \forall j\in N\setminus N^+_v.
\end{equation}
Equations \eqref{eqn:clm:proof:secondAxiomatizationOfSynergyValue:3} and \eqref{eqn:clm:proof:secondAxiomatizationOfSynergyValue:6} imply the correctness of Claim~\ref{clm:proof:secondAxiomatizationOfSynergyValue}.
\end{proof}

Now we are ready to prove the uniqueness statement in Theorem~\ref{thm:secondAxiomatizationOfSynergyValue}. We will do this by showing that, if $x\in\Psi$ is a measure satisfying properties  $P_6$, $P_7$, $P_8$, and $P_{10}$, then $x=\chi$. The proof will be an inductive one over $|N^+_v|$; we will prove that
\begin{equation}\label{eqn:proof:secondAxiomatizationOfSynergyValue:1}
x=\chi, \forall (N,v)\in\G :|N^+_v|=0,
\end{equation}
\noindent and that
\begin{equation}\label{eqn:proof:secondAxiomatizationOfSynergyValue:2}
x=\chi,\ \forall (N,v)\in\G :|N^+_v|<s\ \ \ \ \ \Rightarrow\ \ \ \ \ x=\chi,\ \forall (N,v)\in\G :|N^+_v|=s.
\end{equation}
\ \newline\textbf{Step 1: Proving that \eqref{eqn:proof:secondAxiomatizationOfSynergyValue:1} holds:}

Definition~\ref{def:N^+_v} implies that, for every set of players $N\subset\N$, there exists exactly one game $(N,v)\in\G$ such that $|N^+_v|=0$; this is the game in which every coalition's value equals zero. In other words, it is the game $(N,v^0)$ where $v^0(C)=0, \forall C\subseteq N$; this is the only possible game in which every $i\in N$ is a null player. Since both $x$ and $\chi$ satisfy Property~$P_6$, and since the empty set is assumed to always have zero synergy, we conclude that:
$$
\forall C\subseteq N,\ x^N_{v^0}(C) = \chi^N_{v^0}(C) = 0.
$$
Therefore, Equation~\eqref{eqn:proof:secondAxiomatizationOfSynergyValue:1} holds, which is what we wanted to prove in Step~1.

\ \newline\textbf{Step 2: Proving that \eqref{eqn:proof:secondAxiomatizationOfSynergyValue:2} holds:}

To this end, assume that $x=\chi, \forall (N,v)\in\G :|N^+_v|<s$ for some $s\in{\mathbb N}$, and let $(N,w)$ be a game in $\G$ such that $|N^+_w|=s$. Since $s>0$, there exists at least one player in $N^+_w$. Let $b$ be an arbitrary player in $N^+_w$. To prove that \eqref{eqn:proof:secondAxiomatizationOfSynergyValue:2} holds, it suffices to show that the following equation holds:
\begin{equation}\label{eqn:proof:secondAxiomatizationOfSynergyValue:4}
x^N_w(C) = \chi^N_w(C), \forall C\subseteq N.
\end{equation}
To this end, Claim~\ref{clm:proof:secondAxiomatizationOfSynergyValue} implies that:
$$
\left|N^+_{w_b}\right| < \left|N^+_w\right|.
$$
Our inductive hypothesis then implies that:
\begin{equation}\label{eqn:proof:secondAxiomatizationOfSynergyValue:5}
x^N_{w_b}(C) = \chi^N_{w_b}(C),\ \ \forall C\subseteq N.
\end{equation}
Importantly, for every $i\in N\setminus\{b\}$, we know from Equation~\eqref{eqn:clm:proof:secondAxiomatizationOfSynergyValue:4} that:
$$
\forall C\subseteq N,\ \MC^C_i(N,w_b) = \MC^{C\setminus\{b\}}_i(N,w).
$$
Therefore,
$$
\forall C\subseteq N\setminus\{b\},\ \MC^C_i(N,w_b) = \MC^C_i(N,w).
$$
Based on this, as well as the fact that both $x$ and $\chi$ satisfy Property~$P_{10}$, we find that:
\begin{equation}\label{eqn:proof:secondAxiomatizationOfSynergyValue:6}
\forall C\subseteq N\setminus\{b\},\ x^N_{w_b}(C) = x^N_w(C),
\end{equation}
\begin{equation}\label{eqn:proof:secondAxiomatizationOfSynergyValue:7}
\forall C\subseteq N\setminus\{b\},\ \chi^N_{w_b}(C) = \chi^N_w(C).
\end{equation}
Equations~\eqref{eqn:proof:secondAxiomatizationOfSynergyValue:5}, \eqref{eqn:proof:secondAxiomatizationOfSynergyValue:6} and \eqref{eqn:proof:secondAxiomatizationOfSynergyValue:7} imply that the following holds
\begin{equation}\label{eqn:proof:secondAxiomatizationOfSynergyValue:3}
x^N_w(C) = \chi^N_w(C), \forall C\subseteq N\setminus\{b\},
\end{equation}
Now since $\chi$ satisfies Property~$P_7$, then $(N,w-\chi^N_w)$ is a game of dummies, meaning that there exist real numbers $(y_i)_{i\in N}$ such that:
\begin{equation}\label{eqn:lem:proof:uniquenessOfSyergyValue:7}
\forall C \subseteq N,\ \  w(C) - \chi^N_w(C) = \sum_{i\in C} y_i.
\end{equation}
Likewise, $x$ satisfies Property~$P_7$, and so there exist real numbers $(z_i)_{i\in N}$ such that:
\begin{equation}\label{eqn:lem:proof:uniquenessOfSyergyValue:8}
\forall C \subseteq N,\ \  w(C) - x^N_w(C) = \sum_{i\in C} z_i.
\end{equation}
Now since both $x$ and $\chi$ satisfy Property~$P_8$, then:
$$
\sum_{C\subseteq N}x^N_v(C) = \sum_{C\subseteq N}\psi^N_v(C)
$$
This, as well as equations \eqref{eqn:lem:proof:uniquenessOfSyergyValue:7} and \eqref{eqn:lem:proof:uniquenessOfSyergyValue:8}, imply that:
\begin{equation}\label{eqn:lem:proof:uniquenessOfSyergyValue:9}
\sum_{i\in N} y_i = \sum_{i\in N} z_i.
\end{equation}
We also know that:
\begin{eqnarray}\label{eqn:lem:proof:uniquenessOfSyergyValue:10}
\forall i\in N\setminus\{b\},\ \ y_i &=& w(\{i\}) - \chi^N_w(\{i\}) \eqComment{following Equation~\eqref{eqn:lem:proof:uniquenessOfSyergyValue:7}} \nonumber\\
                                     &=& w(\{i\}) - x^N_w(\{i\}) \eqComment{following Equation~\eqref{eqn:proof:secondAxiomatizationOfSynergyValue:3}} \nonumber\\
                                     &=& z_i. \eqComment{based on Equation~\eqref{eqn:lem:proof:uniquenessOfSyergyValue:8}}
\end{eqnarray}
Equations \eqref{eqn:lem:proof:uniquenessOfSyergyValue:9} and \eqref{eqn:lem:proof:uniquenessOfSyergyValue:10} imply that:
\begin{equation}\label{eqn:lem:proof:uniquenessOfSyergyValue:11}
y_i=z_i,\ \forall i\in N.
\end{equation}
Equations \eqref{eqn:lem:proof:uniquenessOfSyergyValue:7}, \eqref{eqn:lem:proof:uniquenessOfSyergyValue:8} and \eqref{eqn:lem:proof:uniquenessOfSyergyValue:11} imply the correctness of Equation~\eqref{eqn:proof:secondAxiomatizationOfSynergyValue:4}, which in turn implies the correctness of \eqref{eqn:proof:secondAxiomatizationOfSynergyValue:2} as discussed earlier. This concludes step Step 2 and, consequently, concludes the proof of Theorem~\ref{thm:secondAxiomatizationOfSynergyValue}.
$\ \ \ \square$
\end{proof}


\section{A Different Interpretation}\label{sec:AlternativeInterpretation}

As mentioned in the introduction, we interpret synergy as a \textit{deviation from the norm}. Earlier in Section~\ref{sec:FirstDefinition}, to reflect this interpretation, we added Property~$P_8$ to the axiomatization of the synergy measure.\footnote{\footnotesize See our comment on Property~$P_8$ in Section~\ref{sec:FirstDefinition} for more details.} By following that route, we arrived at the Synergy value, $\chi\in\Psi$. In this section, we will arrive at the same measure, but by following an alternative route. Specifically, this route starts by considering the ``norm'' to be the case where every member makes the same impact that he or she \textit{usually makes in a group}. To better understand the intuition behind this, 
%
%
let us revisit our school scenario. By observing the performance of
different groups, the teacher realizes that each one of the shy
students usually (i.e., on average) hardly makes any impact on his or
her group. This observation alone suggests that putting those students
together in a coalition would lead to a low score, unless there was
some positive synergy that somehow improves the group
performance. This could explain why, when that coalition obtained a
score of 70, the synergy was deemed positive---that score was greater
than what was expected (judging from the average impact of each
member). By a similar reasoning, the synergy in the coalition of
leaders was deemed negative---the teacher expected a higher score from
those students (again judging from the average impact of each of
them).

To try and develop a measure that captures the deviation from such a
``norm'', we must first develop a measure that quantifies the
``\textit{average impact}'' of each player, i.e., the \textit{average
  amount of utility that the player contributes towards the value of a
  coalition}. To put it differently, we must first answer the
following fundamental question:

\vspace*{0.5ex}\hspace*{0.35cm}\parbox{12.5cm}{\textit{Judging from
    the values of the different coalitions in a game, how can we
    adequately quantify the average impact that a particular player
    makes?  }}\vspace*{0.5ex}

\noindent This is particularly relevant when all what we know for a
fact are the values themselves, and not the exact factors that may
have influenced those values. Again consider our school scenario,
where the impact that a student has on a group may depend on various
factors, such as the student's reputation at school, the past
experiences that he or she may have shared with other members,
etc. Some of those factors are very hard (if not impossible) to
observe or quantify. On the other hand, the group scores are
observable and comparable. In such cases, one might desire the ability
to measure the average impact of a player, depending solely on the
values of different coalitions.

To this end, a \textit{measure of average impact}, $\theta$, is a
function that assigns to every player $i\in N$ a real value
representing the average impact that $i$ makes on the value of a
coalition containing $i$. As such, if we denote the set of all such
measures by $\Theta$, then every $\theta\in\Theta$ is a real valued
function on:
$$
\bigcup_{(N,v)\in \G} \big\{(N,v,i) : i\in N\big\}.
$$
For notational convenience, we will write $\theta^N_v(i)$ instead of $\theta((N,v,i))$.

It is important to note that $\theta^N_v(i)$ only considers the coalitions that contain $i$. As for any of the remaining coalitions, $i$ clearly has zero impact.\footnote{\footnotesize This is because we restrict our attention to characteristic function games. However, this assumption might not hold in \textit{partition function games}~\cite{Lucas:Thrall:63}, where the value of a coalition may be influenced by the actions of non-members.} This makes $\theta^N_v(i)/2$ the average impact of $i$ \textit{on all coalitions} in the game.

The following are seemingly reasonable properties to have in a measure of average impact $\theta\in\Theta$:

\begin{itemize}
    \item[] $P_{11}$. (\textbf{Symmetric-Impact}): A measure $\theta\in\Theta$ satisfies $P_{11}$ if, for every game $(N,v)\in\G$, and every pair of symmetric players $i$ and $j$ in that game, we have: $\theta^N_v(i)=\theta^N_v(j)$;\smallskip
    \item[] $P_{12}$. (\textbf{Average-Value}): A measure $\theta\in\Theta$ satisfies $P_{12}$ if, for every game $(N,v)\in\G$, we have: $\frac{1}{2}\sum_{i\in N}\theta^N_v(i) = 2^{-n}\sum_{C\subseteq N}v(C)$;\smallskip
    \item[] $P_{13}$. (\textbf{Marginal-Impact}): A measure $\theta\in\Theta$ satisfies $P_{13}$ if, for every pair of games $(N,v)$ and $(N,w)$ with the same set of players, and for every player $i\in N$, if
    \begin{equation}\label{eqn:marginalityInAverageImpace}
    \forall C\subseteq N\setminus \{i\},\ \ \MC^C_i(N,v) = \MC^C_i(N,w),
    \end{equation}
    then the following holds:
    $$
    \theta^N_v(i) = \theta^N_w(i).
    $$
\end{itemize}

The intuition behind property $P_{11}$ is rather clear. As for $P_{12}$, it also seems intuitive as it states that the average outcome of a game is obtained when every player in that game makes its average impact. Finally, to comment on $P_{13}$, this property implies that if the performance of a certain player remained unchanged, then the average impact of this player should also remain unchanged, regardless of whether the performance of some other player(s) in the game has changed.


Inspired by the above properties, we propose a measure of average impact, called the \textit{Average-Impact value}, and defined as follows:
\begin{definition}[\textbf{Average-Impact (AI) Value}]\label{def:AverageImpactValue}
The \emph{Average-Impact (AI) Value} is the measure of average impact, $\lambda\in\Theta$, defined for every game $(N,v)\in\G$ and every player $i\in N$ as follows:
\begin{equation}
\lambda^N_v(i) = \overline{\phi}_i(N,v).
\end{equation}
\end{definition}
The following theorem implies that properties $P_{11}$, $P_{12}$ and $P_{13}$ axiomatize $\lambda$.

\begin{theorem}\label{thm:uniquenessOfAverageImpact}
The AI value is the only measure in $\Theta$ that satisfies properties $P_{11}$, $P_{12}$, and $P_{13}$.
\end{theorem}

\begin{proof}
We begin by proving that the AI value satisfies the three properties listed in the statement of the theorem.

\begin{claim}\label{clm:AIsatisfiesProperties}
The AI value satisfies properties $P_{11}$, $P_{12}$ and $P_{13}$.
\end{claim}

\begin{proof}
Let $(N,v)$ be a game in $\G$. We already showed that Equation~\eqref{proof:P5:2} holds for every pair of symmetric players $i,j\in N$, meaning that $\lambda$ satisfies Property~$P_{11}$. Moving on to $P_{12}$, we have:
\begin{eqnarray}
\frac{1}{2}\sum_{i\in N}\overline{\phi}_i(N,v)  &=&  2^{-n} \sum_{i\in N}\sum_{C\subseteq N : i\in C} \phi_i(C,v) \eqComment{following Equation~\eqref{eqn:averageShapleyValue}} \nonumber\\
        &=& 2^{-n} \sum_{C\subseteq N}\sum_{i\in C} \phi_i(C,v)  \nonumber\\
        &=& 2^{-n} \sum_{C\subseteq N}v(C)  \eqComment{as $\phi$ satisfies $P_3$} \nonumber
\end{eqnarray}
Consequently, $\lambda$ satisfies $P_{12}$. Finally, to prove that $\lambda$ satisfies $P_{13}$, it suffices to note that, for every $i\in N$, $\overline{\phi}_i(N,v)$ is by definition a weighted average of the marginal contributions of $i$ to all the coalitions in the game $(N,v)$ (see equations \eqref{eqn:shapleyValue} and \eqref{eqn:averageShapleyValue}).
\end{proof}

Having proved the correctness of Claim~\ref{clm:AIsatisfiesProperties}, it remains to prove the uniqueness statement in Theorem~\ref{thm:uniquenessOfAverageImpact}. For this, we need to introduce two additional definitions:
\begin{definition}[$\C(N,v)$]\label{def:C^N_v}
For every game $(N,v)\in\G$, the set $\C(N,v) \subseteq 2^N$ is the set of coalitions in $(N,v)$ that each have at least one subset whose value is non-zero. Formally,
\begin{equation}\label{eqn:C^N_v}
\C(N,v) = \{C\subseteq N: \exists C'\subseteq C, v(C')\neq 0\}.
\end{equation}
\end{definition}
\begin{definition}[$(N,v_S)$]\label{def:(N,v_S)}
For every set of players $N\subset\N$, and every coalition, $S\subseteq N$, and every characteristic function $v:2^N \to {\mathbb R}$, the game $(N,v_S)$ is defined as follows:
\begin{equation}\label{eqn:v_S}
v_S(C) = v(C) - v(C\cap S),\ \ \forall C\subseteq N.
\end{equation}
\end{definition}
Now, we are ready to prove the uniqueness statement in Theorem~\ref{thm:uniquenessOfAverageImpact}. More specifically, assuming that $\delta\in\Theta$ is an average-impact measure satisfying $P_{11}$, $P_{12}$ and $P_{13}$, we will prove that $\delta = \lambda$. The proof will be an inductive one over $|\C(N,v)|$. More specifically, we will prove that
\begin{equation}\label{eqn:proof:uniquenessOfAverageImpact:1}
\delta=\lambda,\ \ \forall (N,v)\in\G :|\C(N,v)|=0,
\end{equation}
\noindent and that
\begin{equation}\label{eqn:proof:uniquenessOfAverageImpact:2}
\delta=\lambda,\ \forall (N,v)\in\G :|\C(N,v)|<s\ \ \ \ \ \Rightarrow\ \ \ \ \ \delta=\lambda,\ \forall (N,v)\in\G :|\C(N,v)|=s.
\end{equation}

\ \newline\textbf{Step 1: Proving that \eqref{eqn:proof:uniquenessOfAverageImpact:1} holds:}

Definition~\ref{def:C^N_v} implies that, for every set of players $N\subset\N$, there exists exactly one game $(N,v)\in\G$ for which $|\C(N,v)|=0$; this is the game in which every coalition's value is zero, i.e., it is the game $(N,v^0)$, where $v^0(C)=0, \forall C\subseteq N$. In this game, every pair of players are symmetric. Consequently, the following holds, as both $\delta$ and $\lambda$ satisfy Property~$P_{11}$:
\begin{equation}\label{eqn:proof:uniquenessOfAverageImpact:3}
\delta^N_{v^0}(i) = \delta^N_{v^0}(j),\ \ \forall i,j\in N,
\end{equation}
\begin{equation}\label{eqn:proof:uniquenessOfAverageImpact:4}
\lambda^N_{v^0}(i) = \lambda^N_{v^0}(j),\ \ \forall i,j\in N.
\end{equation}
Furthermore, since both $\delta$ and $\lambda$ satisfy Property~$P_{12}$, then:
\begin{equation}\label{eqn:proof:uniquenessOfAverageImpact:5}
\sum_{i\in N}\delta^N_{v^0}(i) = \sum_{i\in N}\lambda^N_{{v^0}}(i).
\end{equation}
Equations \eqref{eqn:proof:uniquenessOfAverageImpact:3}, \eqref{eqn:proof:uniquenessOfAverageImpact:4} and \eqref{eqn:proof:uniquenessOfAverageImpact:5} imply that $\delta^N_{v^0}=\lambda^N_{v^0}$, meaning that Equation \eqref{eqn:proof:uniquenessOfAverageImpact:1} holds, which is what we wanted to prove in Step~1.

\ \newline\textbf{Step 2: Proving that \eqref{eqn:proof:uniquenessOfAverageImpact:2} holds:}

Assuming that the following holds:
\begin{equation}\label{eqn:proof:uniquenessOfAverageImpact:6}
\forall (N,v)\in\G :|\C(N,v)|<s,\ \ \forall i\in N,\ \ \delta^N_v(i) = \lambda^N_v(i),
\end{equation}
we need to prove that the following holds for every $i\in N$, where $w$ is a characteristic function such that $|\C(N,w)|=s$:
\begin{equation}\label{eqn:proof:uniquenessOfAverageImpact:7}
\delta^N_w(i) = \lambda^N_w(i).
\end{equation}
Importantly, for every game $(N,v)\in\G$, and every $S\in\C(N,v)$, we know that $\C(N,v_S) \subset \C(N,v)$ \cite{Young:85}. This implies that $|\C(N,w_S)| < |\C(N,w)|$, meaning that $|\C(N,w_S)| < s$. Thus, based on our assumption that Equation~\eqref{eqn:proof:uniquenessOfAverageImpact:6} holds, we find that:
\begin{equation}\label{eqn:proof:uniquenessOfAverageImpact:8}
\forall i\in N,\ \ \delta^N_{w_S}(i) = \lambda^N_{w_S}(i).
\end{equation}
Now, denote by $\widehat{C}$ the set formed by the intersection of all coalitions in $\C(N,w)$, i.e.,
\begin{equation}\label{eqn:proof:uniquenessOfAverageImpact:9}
\widehat{C} = \bigcap_{C\in \C(N,w)} C.
\end{equation}
We will prove the correctness of Equation~\eqref{eqn:proof:uniquenessOfAverageImpact:7} in two steps, based on $\widehat{C}$. Specifically, in Step~2.1 we will prove that the equation holds for every $i\in N\setminus\widehat{C}$, while in Step~2.2 we will prove that it holds for every $i\in\widehat{C}$.

\ \newline\textbf{Step 2.1.} For every $i\in N\setminus\widehat{C}$, there exists a coalition in $\C(N,w)$ that does not contain $i$. Let $S$ be one such coalition. Then, for every $C\subseteq N$, we have:
\begin{eqnarray}
\MC^C_i(N,w_S) &=& w_S(C\cup\{i\}) - w_S(C) \eqComment{based on Definition~\ref{def:marginalContribution}} \nonumber\\
               &=& w(C\cup\{i\}) - w((C\cup\{i\})\cap S) - w(C) + w(C\cap S) \eqComment{based on Definition~\ref{def:(N,v_S)}} \nonumber\\
               &=& w(C\cup\{i\}) - w(C\cap S) - w(C) + w(C\cap S) \eqComment{because $i\notin S$} \nonumber\\
               &=& w(C\cup\{i\}) - w(C) \nonumber\\
               &=& \MC^C_i(N,w) \eqComment{based on Definition~\ref{def:marginalContribution}} \nonumber
\end{eqnarray}
Consequently, the following two equations hold, as both $\delta$ and $\lambda$ satisfy Property~$P_{13}$:
\begin{equation}\label{eqn:proof:uniquenessOfAverageImpact:10}
\forall i\in N\setminus\widehat{C},\ \ \delta^N_{w_S}(i) = \delta^N_w(i),
\end{equation}
\begin{equation}\label{eqn:proof:uniquenessOfAverageImpact:11}
\forall i\in N\setminus\widehat{C},\ \ \lambda^N_{w_S}(i) = \lambda^N_w(i).
\end{equation}
Equations \eqref{eqn:proof:uniquenessOfAverageImpact:8}, \eqref{eqn:proof:uniquenessOfAverageImpact:10} and \eqref{eqn:proof:uniquenessOfAverageImpact:11} imply that Equation~\eqref{eqn:proof:uniquenessOfAverageImpact:7} holds for every $i\in N\setminus\widehat{C}$, which is what we wanted to prove in Step~2.1.

\ \newline\textbf{Step 2.2.} In this step, we want to prove that Equation~\eqref{eqn:proof:uniquenessOfAverageImpact:7} holds for every $i\in\widehat{C}$. In the case where $\widehat{C}=\emptyset$, the conclusion follows vacuously. Next, we will deal with the case where $|\widehat{C}|=1$, and then deal with the case where $|\widehat{C}|>1$. In both cases, we will make use of the fact that the following holds, as both $\delta$ and $\lambda$ satisfy Property~$P_{12}$: 
\begin{equation}\label{eqn:proof:uniquenessOfAverageImpact:12}
\sum_{i\in N}\delta^N_w(i) = \sum_{i\in N}\lambda^N_w(i).
\end{equation}
First, assume that $|\widehat{C}|=1$, and let $b$ denote the only player in $\widehat{C}$. Since we know that Equation~\eqref{eqn:proof:uniquenessOfAverageImpact:7} holds for every $i\in N\setminus\{b\}$, then Equation~\eqref{eqn:proof:uniquenessOfAverageImpact:12} implies that $\delta^N_w(b) = \lambda^N_w(b)$, i.e., it implies that Equation~\eqref{eqn:proof:uniquenessOfAverageImpact:7} holds for the one player in $\widehat{C}$, which is what we wanted to show.

Now, assume that $|\widehat{C}|>1$. Definition~\ref{def:C^N_v} and Equation~\eqref{eqn:proof:uniquenessOfAverageImpact:9} imply that every coalition $C$ that does not contain $\widehat{C}$ satisfies: $w(C)=0$. This implies that, for every pair of players, $i,j\in \widehat{C}$, and every $C\subseteq N\setminus\{i,j\}$, we have: $w(C\cup\{i\}) = w(C\cup\{j\}) = 0$. Therefore, every pair of players in $\widehat{C}$ are symmetric in the game $(N,w)$. Consequently, the following holds, as both $\delta$ and $\lambda$ satisfy Property~$P_{11}$:
\begin{equation}\label{eqn:proof:uniquenessOfAverageImpact:13}
\delta^N_{w}(i) = \delta^N_{w}(j),\ \ \forall i,j\in \widehat{C},
\end{equation}
\begin{equation}\label{eqn:proof:uniquenessOfAverageImpact:14}
\lambda^N_{w}(i) = \lambda^N_{w}(j),\ \ \forall i,j\in \widehat{C}.
\end{equation}
Recall that we proved in Step~2.1 that Equation~\eqref{eqn:proof:uniquenessOfAverageImpact:7} holds for every $i\in N\setminus\widehat{C}$. Based on this, equations \eqref{eqn:proof:uniquenessOfAverageImpact:12}, \eqref{eqn:proof:uniquenessOfAverageImpact:13} and \eqref{eqn:proof:uniquenessOfAverageImpact:14} imply that Equation~\eqref{eqn:proof:uniquenessOfAverageImpact:7} also holds for every $i\in\widehat{C}$, which is what we wanted to prove. This concludes Step~2.2, and so concludes the proof of Theorem~\ref{thm:uniquenessOfAverageImpact}.
$\ \ \ \square$
\end{proof}

Having proposed $\lambda\in\Theta$ as a measure of the average impact of a player in a game $(N,v)\in\G$, let us now move back to the main focus of this paper---measuring the synergy in a coalition $C\subseteq N$. As mentioned earlier, we interpret synergy as a deviation from the ``norm'', which is considered in this section to be the case where every member makes its average impact. With this interpretation, every measure of average impact, $\theta\in\Theta$, leads to a different measure of synergy, which is: $v(C) - \sum_{i\in C}\theta^N_v(i)$. Now if we adopt $\lambda$ as our measure of choice for the average impact, we arrive at the following measure of synergy: $v(C) - \sum_{i\in C}\lambda^N_v(i)$, which is the same as $\chi$---the measure we arrived at in Section~\ref{sec:FirstDefinition}.


\section{Conclusions}\label{sec:conclusion}

\noindent   In this paper we proposed a new measure of synergy 
in teams and formally analysed its properties.


\bibliographystyle{aaai}
\bibliography{bibliography_new}


\end{document}